\pdfoutput=1
\documentclass[%
 jmp,%
 amsmath,amssymb,
 reprint,%
showkeys,
aps,prl
]{revtex4-1}

\usepackage{graphicx}
\usepackage{dcolumn}
\usepackage{bm}
\usepackage{mathtools}
\mathtoolsset{showonlyrefs=true}
\usepackage{amsmath}
\usepackage{amssymb}
\usepackage{amsthm}
\newtheorem{lemma}{Lemma}

\begin{document}

\preprint{AIP/123-QED}

\title[Swarm behavior of traders]{Swarm behavior of traders \\with different subjective predictions in the Market
}

\author{Hiroshi Toyoizumi}
\affiliation{ Graduate School of Accountancy, Waseda University.
}
 \email{toyoizumi@waseda.jp}
\altaffiliation[Also at ]{Dept. of Applied Mathematics, Waseda University.}


\date{\today}

\begin{abstract}
A combination of a priority queueing model and mean field theory shows the emergence of traders' swarm behavior, even when each has a subjective prediction of the market driven by a limit order book.  Using a nonlinear Markov model, we analyze the dynamics of traders who select a favorable order price taking into account the waiting cost incurred by others.  We find swarm behavior emerges because of the delay in trader reactions to the market, and the direction of the swarm is decided by the current market position and the intensity of zero-intelligent random behavior, rather than subjective trader predictions.
\end{abstract}

\pacs{Valid PACS appear here}
\maketitle

Stock markets use a double auction system with a limit order book (LOB) \cite{doi:10.1080/14697688.2013.803148}, where traders place their sell and buy orders with a specific price (limit orders), as well as their orders with no specific price (market  orders), which executed against the most favorable limit orders.  Because the LOB is the key micro-structure mechanism of the stock market, there are numerous mathematical and empirical analyses of the LOB, which mainly try to understand the observed order distribution on the LOB and the price formation process based on some assumptions of trader behavior \citep{Smith:2003rt,Ichiki:2015aa,doi:10.1080/14697688.2013.803148,Rosu01112009,doi:10.1080/01621459.2014.982278,PhysRevE.64.056136} and possibly to derive the optimal order-placing strategy \citep{maglaras2015optimal,Toyoizumi:2016xe}.  Some of them used Markov processes including queueing model to describe the behavior of the LOB \cite{Cont:2010uq,kelly2015markov,doi:10.1080/14697688.2014.963654}.  At the same time, as Keynes famous quote about stock markets as a beauty contest \cite{keynes1937general}, we need to take into account the higher-order expectation of other traders' belief (the expectation of what others predict the belief of others) to understand the market dynamics \cite{Allen:2006fk}.

Here, we use the concept of mean field game theory \cite{Carmona:2013fk,carmona2015, Lasry:qy, lachapelle2013efficiency, gueant2011mean,PhysRevE.64.056136} incorporating the higher-order belief of traders as Picard iteration (\cite{Dobrushkin:2014rt} for example), and analyze the behavior of traders who act to maximize their rewards by placing their orders in the LOB modeled by a priority queue \citep{hassin2003queue}.  Specifically, we study the effect of the delay of traders’ reactions with a non-linear Markov process \cite{kolokoltsov2010nonlinear}, which leads to the emergence of swarm behavior, and we show that the swarm depends on the current market position and the intensity of random behavior rather than subjective trader predictions.

\begin{figure}[tbp]
\begin{center}
\includegraphics[width=0.5\textwidth]{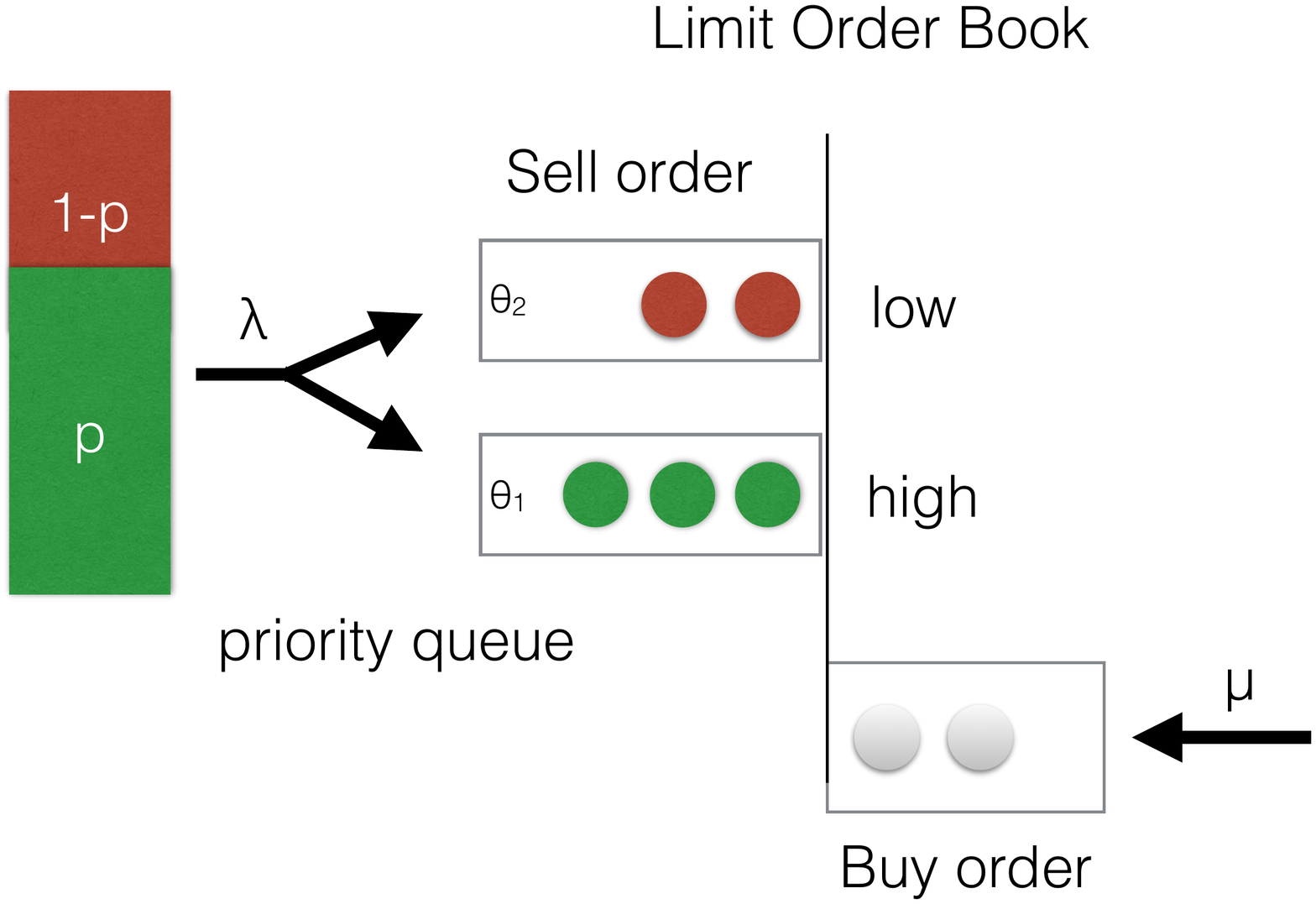}
\caption{A priority queue modelling  of a one-sided LOB. Traders select the prices of their sell orders ($\theta_{1} < \theta_{2}$).  The $\theta_{1}$ queue is executed with the high priority.  Sell and Buy orders arrive as independent Poisson processes with the rate $\lambda$ and $\mu$, and the ratio of $\theta_{1}$-sell orders is $p$. }
\label{fig:priorityQueue.pdf}
\end{center}
\end{figure}

A simplified one-sided LOB \cite{doi:10.1080/14697688.2014.963654} can be modeled by a priority queue \citep{hassin2003queue} (see Figure \ref{fig:priorityQueue.pdf}).  Assume that there is a population of traders that wants to sell stocks.  Traders place their sell orders selecting the prices either $\theta_{1}$ or $\theta_{2}$ ($\Delta \theta = \theta_{2} - \theta_{1} >0$, generally $\Delta \theta$ is fixed to a small value) in the LOB.  The sell orders with different prices will be put into two queues, and wait to be executed with a matching market buy order. The lower-price queue denoted by $\theta_{1}$ is given the priority, and each queue is served on a first-come-first-serve basis.  

We assume that limit sell orders and market buy orders arrive at the LOB as independent Poisson processes with the rate  $\lambda$ and the rate $\mu$, respectively.  Let $p$ be the ratio of $\theta_{1}$ sell orders.   An arriving market buy order will be cleared with a limit sell order at the head of the priority $\theta_{1}$ queue, if there is one.  If there is no sell order in the $\theta_{1}$ queue, the buy order is executed with the sell order at the head of $\theta_{2}$.  When there is no order in either queue, the buy order will be canceled.  We consider a buy-dominant market, which implies $\rho = \lambda/\mu <1$.  Thus, all sell orders will be executed eventually.

When they place a sell order, traders consider its cost for waiting time until execution, which is denoted by $c$ per unit time, and their reward $R_{i}$ when selecting the price $\theta_{i}$ is: 
\begin{align}
R_{i} = \theta_{i} - c W_{i}, \quad i=1,2,
\end{align}
where $W_{i}$ is the waiting time until the $\theta_{i}$ order is executed.   By placing a $\theta_{1}$ order, traders expect their orders to be sold more quickly than $\theta_{2}$ orders and reduce the waiting cost, at the expense of selling them at a lower price.

Traders optimize their behavior according to the waiting cost incurred by the behavior of other traders.  Here are some intuitions: (1) a trader should select $\theta_{1}$ if many others select it, or his or her order will be delayed by $\theta_{1}$ orders, however (2) when others do not select $\theta_{1}$, his or her order will not be severely-delayed and selecting $\theta_{2}$ is better.   Thus, the decision should be affected by cost structure as well as the prediction of other traders' decisions. 

Consider a trader Alice who estimates the waiting cost to sell her stock over $[0,T]$ to make the decision at time $0$.  Let $p$ be the ratio of traders selecting $\theta_{1}$ in the market.  Because orders are executed rapidly, Alice can assume her orders to enter the priority queue of the LOB in the stationary state.  By the conservation law of the workload of priority queues (\citep{hassin2003queue} and Supporting Material \ref{Expected Waiting Times}),  we have the following estimate for the expected waiting times:
\begin{align}
E[W_{1}]&=\frac{1}{(\mu - \lambda p)},\\
E[W_{2}]&= \frac{\mu}{(\mu-\lambda)(\mu-\lambda p)},
\end{align}
in the stationary priority queue.   If many traders select the higher priority ($\theta_{1}$) orders (large $p$), then the waiting times to be executed become large in the both queues, but the orders with the lower priority ($\theta_{2}$ orders) suffer more severely, since $\rho = \lambda/\mu <1$.

Define the value $g$ as the expected gain from selecting the lower price $\theta_{1}$:
\begin{align}\label{eq:gain of 1}
g(p,c) &= E[R_{1}] - E[R_{2}]\\
&= (\theta_{1}-cE[W_{1}])-(\theta_{2}-cE[W_{2}])\\
&= \frac{\rho c/\mu}{(1-\rho)(1-\rho p)} - \Delta \theta,
\end{align}
which is an increasing function of $p$.  Note that the $g$-value is negative when $c$ is small, while $g$ is positive when $p$ and $\rho$ are large.

Figure \ref{fig:Contour_g.pdf} illustrates the typical $g$-value on the $(p,c)$-plane in a slightly buy-dominant market ($\rho = 0.9$), which shows that the traders will change their behavior according to the behavior of other traders.  At the point $A$ in the $(p,c)$-plane, the cost of waiting is large (positive $g$-value) and Alice selects $\theta_{1}$ to increase her reward.  On the other hand, at the point $B$, the cost of waiting is negligible and Alice gains more reward by selecting $\theta_{2}$.  The point $C$ is on the boundary $g(p,c)=0$, and Alice's selection will not affect her reward.  This boundary is an unstable Nash-equilibrium \citep{hassin2003queue}.  

\begin{figure}[htbp]
\begin{center}
\includegraphics[width=0.5\textwidth]{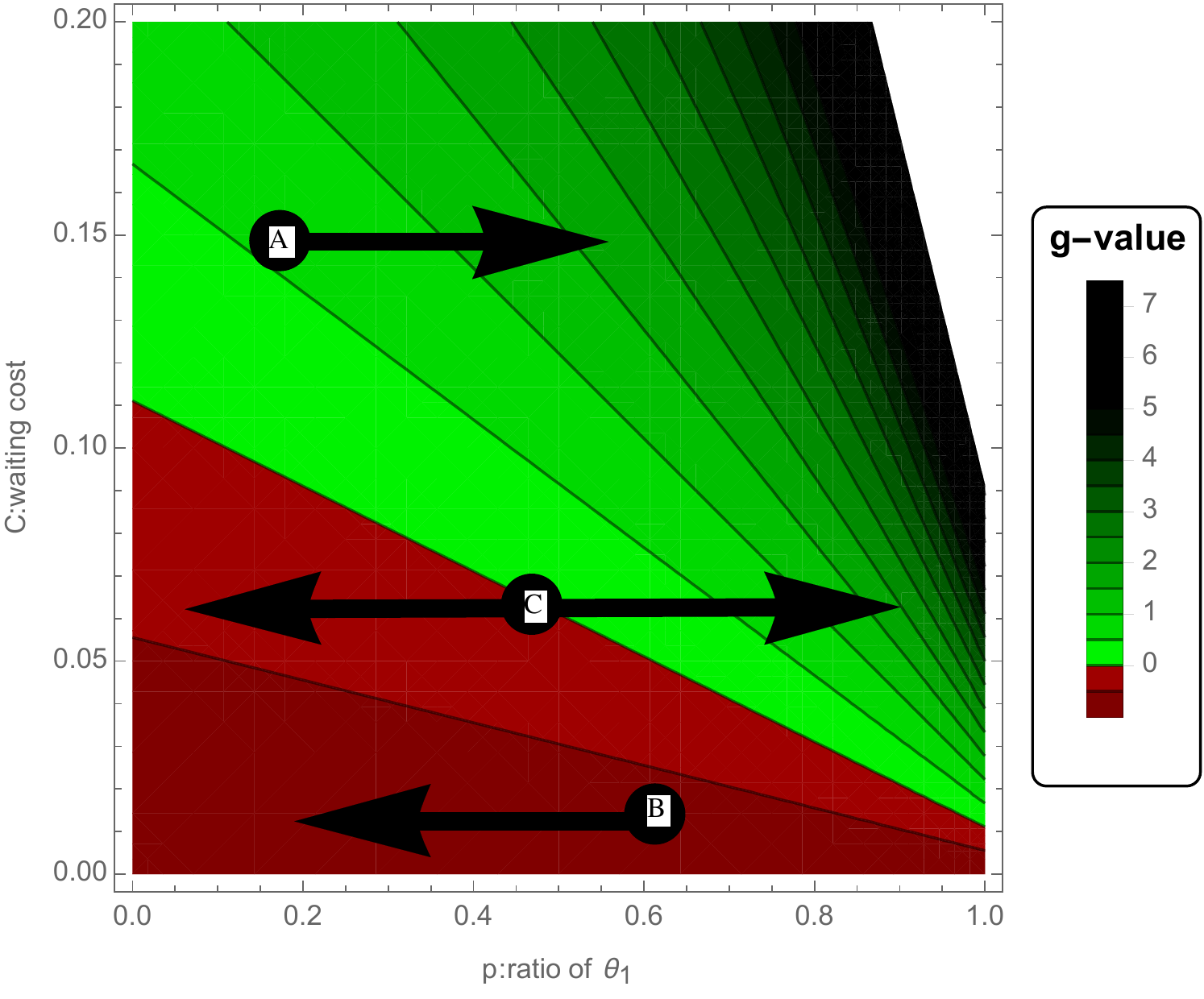}
\caption{The g-value of $g(p,c)$, which is the expected gain to select $\theta_{1}$.  Here we set $\Delta \theta =1$ and $\rho = 0.9$.  The line $\{p=0, g<0\}$ on the red region and $\{p=1,g>0\}$ on the green region are both stable equilibrium, while the boundary $g(p,c)=0$ is unstable.}
\label{fig:Contour_g.pdf}
\end{center}
\end{figure}

Thus, Alice can decide her strategy according to the $g$-value given other traders' behavior $p$ and the cost structure $c$.  However, in general, she should predict other traders' behavior $p$ in the future market, because it will take some time for her to implement her strategy.  

Suppose that at time $0$ Alice has an arbitrary initial prediction of the future market ratio of selecting $\theta_{1}$ as $x^{a}=(x^{a}(t))_{t\in [0,T]}$ given the known current ratio $x^{a}(0)=p_{0}$.  Now Alice picks a trader Bob and analyzes his behavior through  a thought experiment assuming the market (traders other than Bob) follows her prediction.   Alice also assumes that Bob changes his behavior according to the two-state non-homogeneous Markov process with the infinitesimal generator:
\begin{align}
Q[x^{a}(t)] = \left(\begin{array}{cc}
-\alpha_{1}(x^{a}(t)) & \alpha_{1}(x^{a}(t)) \\
 \alpha_{2}(x^{a}(t)) & -\alpha_{2}(x^{a}(t))
\end{array}\right)
\end{align}
and the master (Kolmogorov) differential equation:
\begin{align}\label{eq:ODE}
\frac{dx_{1}^{a}(t)}{dt}=-\alpha_{1}(x^{a}(t))x_{1}^{a}(t)+\alpha_{2}(x^{a}(t))(1-x_{1}^{a}(t)),
\end{align}
where $x_{1}^{a}(t)= P\{\text{Bob selects $\theta_{1}$ at time $t$}\}$, and $\alpha_{1}(p)$ and $\alpha_{2}(p)$ are functions of Alice's prediction of market $p=x^{a}(t)$ and defined by 
\begin{align}\label{eq:alpha}
\alpha_{1}(p)= \beta + \alpha g^{-}(p,c), \\
\alpha_{2}(p) = \beta + \alpha g^{+}(p,c).
\end{align}
Here, $g^{+} (p,c) = \max(0,g(p,c))$ and $g^{-} (p,c) = -\min(0,g(p,c))$, and $\alpha$ and $\beta$ are some non-negative constants.  The instantaneous transition rate $\alpha_{1}(p)$ and $\alpha_{2}(p)$ are interpreted as follows.   Under the market $p$, Bob changes his selection with the rate proportional to the gain obtained if he changed his selection.  The term $\beta\geq 0$ represents the possibility of random behavior change irrelevant to his expected gain (zero-intelligence behavior \cite{doi:10.1080/14697688.2013.803148,Toyoizumi:2016xe}).  Using this Markov chain model, Alice takes into account that Bob changes his strategy quickly if he gains more but he always needs time to adapt to the market change. 

Now Alice expects that Bob select his order price according to the probability $x^{a}_{1}$.  Because other traders and the market as a whole change their strategies just as Bob does, she must change her prediction of the market to $x^{a}_{1}$ from $x^{a}$.  With her new market prediction $x^{a}_{1}$, Alice calibrates her prediction about Bob's behavior again based on the marker $x^{a}_{1}$, and then she obtains a newer prediction $x^{a}_{2}$ of Bob's behavior.  She can repeat this procedure and get her $n$-th prediction $x^{a}_{n}$ based on the market $x_{n-1}^{a}$ satisfying
\begin{align}\label{eq:ODE Picard Iteration}
\frac{dx_{n}^{a}(t)}{dt}=-\alpha_{1}(x_{n-1}^{a}(t))x_{n}^{a}(t)+\alpha_{2}(x_{n-1}^{a}(t))(1-x_{n}^{a}(t)),
\end{align}
with the known initial condition $x_{n}^{a}(0)=p_{0}$.  Using the argument similar to Picard iteration (\cite{Dobrushkin:2014rt} for example), we can show that 
\begin{align}\label{eq:the limit}
x^{a}_{n}(t) \to x(t),
\end{align}
uniformly on the finite interval $[0,T]$ as $n \to \infty$ for an unique smooth function $x$ (see Supplemental Material \ref{Picard-like Iteration of Market Prediction} for the detail).  Taking $n\to\infty$ in \eqref{eq:ODE Picard Iteration}, we can show that $x$ satisfies
\begin{align}\label{eq:ODEofFixedPoint}
\frac{dx(t)}{dt}=-\alpha_{1}(x(t))x(t)+\alpha_{2}(x(t))(1-x(t)).
\end{align}
with the initial condition $x(0)=p_{0}$, which is a two-state nonlinear Markov process \cite{kolokoltsov2010nonlinear}.

Generally, nonlinear differential equations such as \eqref{eq:ODEofFixedPoint} does not necessarily have a unique solution.  However, the limit $x$ is irrelevant to Alice's initial expectation $x^{a}$, but only depends on the current market position $p_{0}$ because of the Lipschitz continuity of $\alpha_{1}$ and $\alpha_{2}$ (see Supplemental Material).  Further, another trader Charlie, who may have a different initial prediction $x^{c}$ in the market, yet shares the same information about the current market $p_{0}$, must reach the same conclusion as Alice (this is the reason we dropped the superscript $a$ of $x$ in \eqref{eq:the limit}).  Thus, Alice, Bob, Charlie and all other traders reach the same limit prospect $x$ of the future market, and then all traders act similarly according to their own interests based on the future market $x$, and a swarm behavior of traders emerges, even though each has the different initial subjective predictions on the market. 

Figure \ref{fig:PicardIterationDecreasing} - \ref{fig:PicardIterationLOB-Equi.pdf} illustrate examples how the limit behavior achieved by updating traders' expectations in the case when the random switching effect is small ($\beta=0.1$).  There, $p_{e}$ is the critical ratio of the market  that satisfies $g(p,c)=0$.   As seen in Figure \ref{fig:Contour_g.pdf}, if $p$ is greater than $p_{e}$, it is better to select the lower price $\theta_{1}$.    We only show the result when there is a non-stable equilibrium $p_{e}$ in the market $(0<p_{e}<1)$, as otherwise trader's selections are obvious.  All examples show the convergence of updates $x_{n}$ to the limit $x$, and thus swarm behaviors emerge.

\begin{figure}[tbp]
\begin{center}
\includegraphics[width=0.45\textwidth]{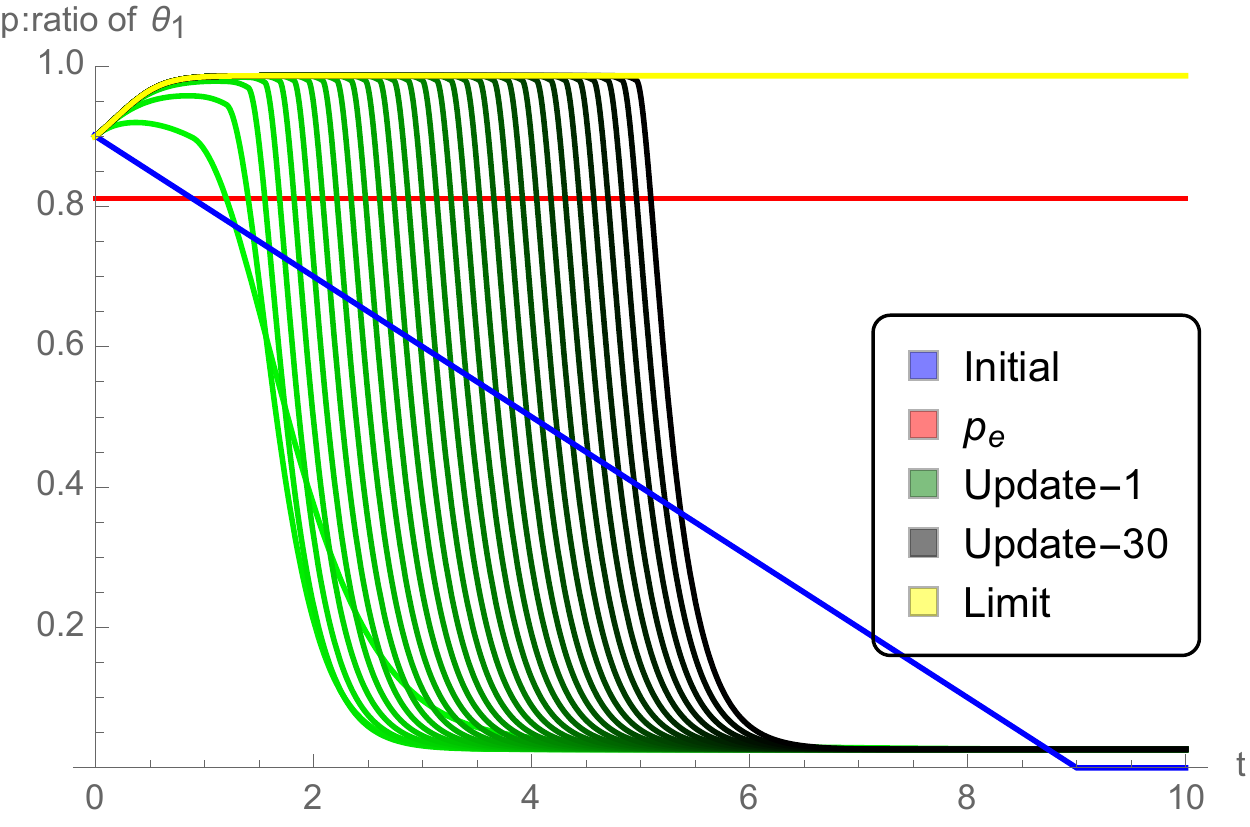}
\includegraphics[width=0.45\textwidth]{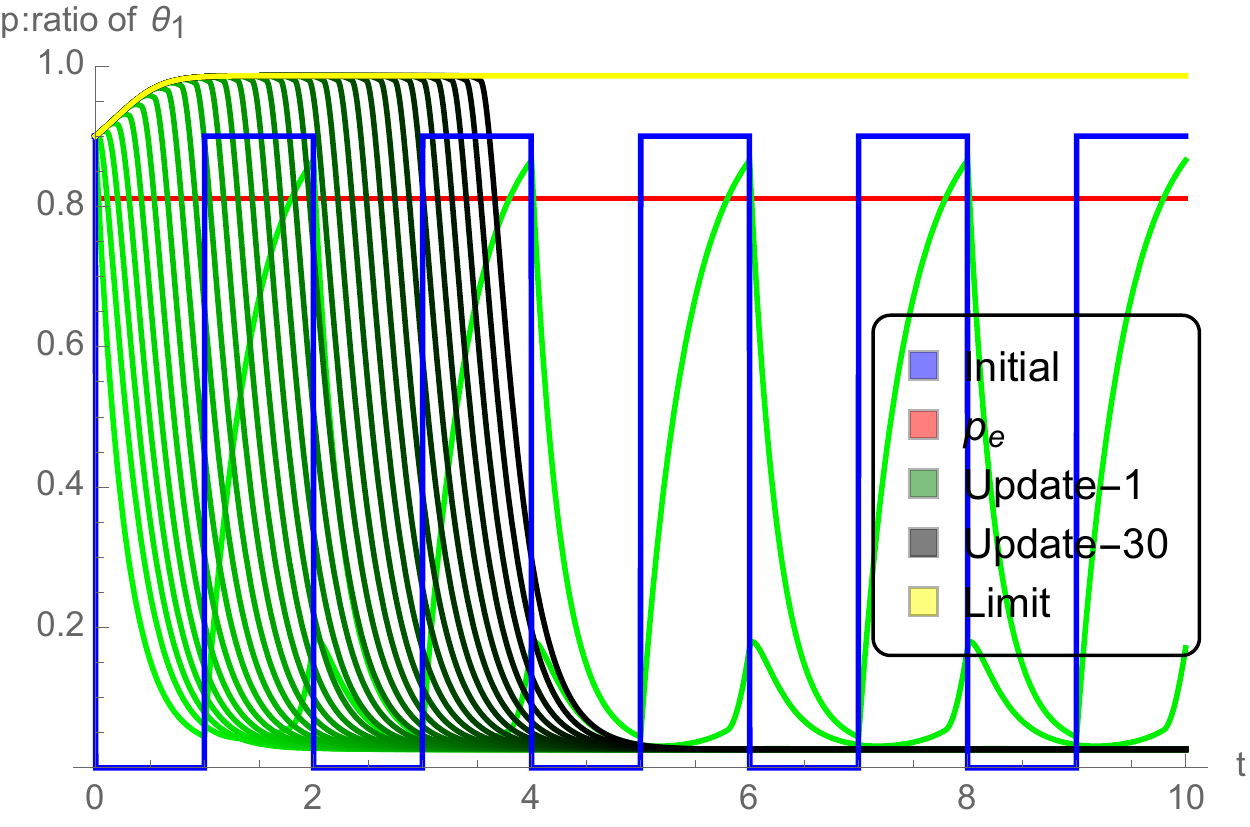}
\includegraphics[width=0.45\textwidth]{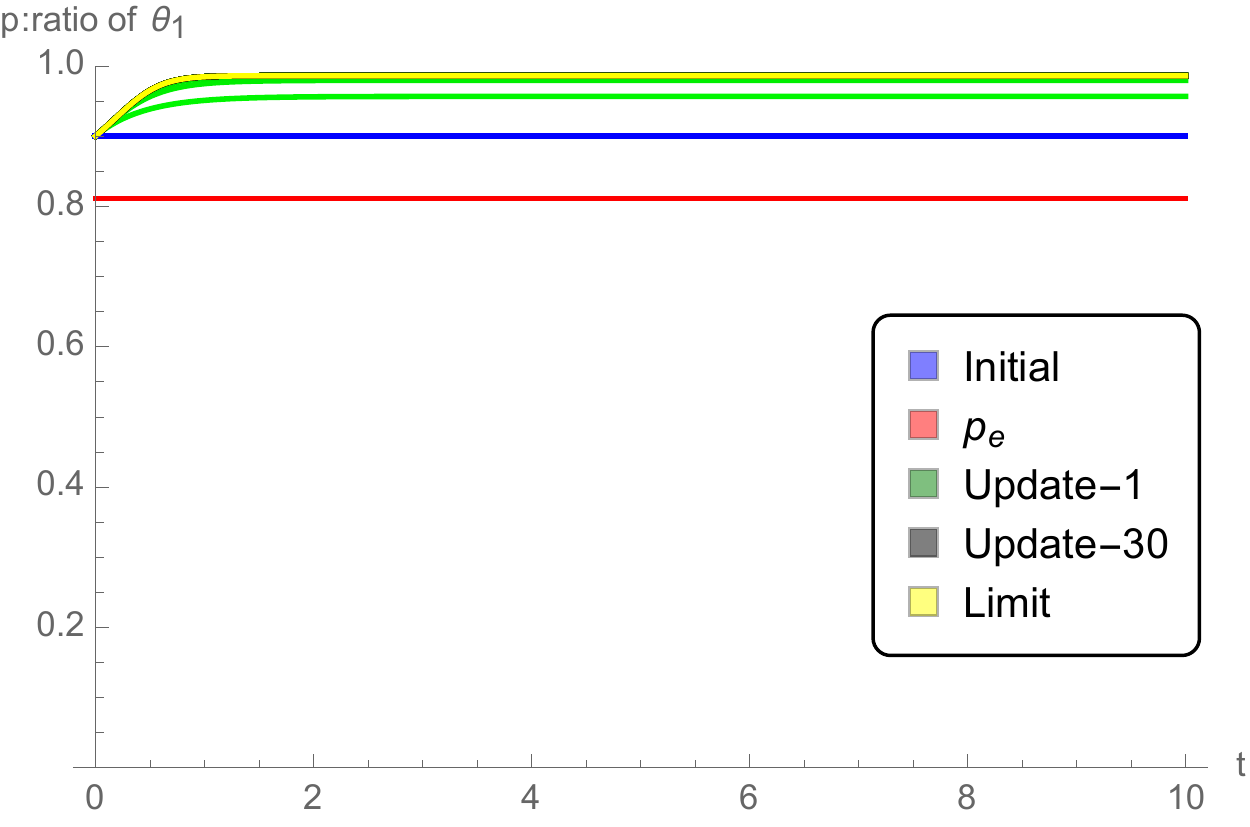}
\caption{Example of the market prediction by three traders starting from the current market ratio $p_{0}=0.9$.  The first 30 updates of predictions (green lines with gradation) are shown with the different traders (blue lines); decreasing, oscillating step and static function.  All have different initial predictions, but eventually they are corrected to the unique limit (yellow line) satisfying \eqref{eq:ODEofFixedPoint}. Here we set $\Delta \theta =1, \rho = 0.9, c = 0.03, \alpha = 5, \beta = 0.1$.  The critical ratio is $p_{e} = 0.811111$.  }
\label{fig:PicardIterationDecreasing}
\end{center}
\end{figure}

Figure \ref{fig:PicardIterationDecreasing} shows an example of three different initial predictions of traders: decreasing (Alice), oscillating step (Bob) and static (Charlie), starting from the same current market $p_{0}=0.9$, which is higher than the critical value $p_{e}$ just like the position $A$ in Figure \ref{fig:Contour_g.pdf}.  The repeated updates of prediction (green lines with gradation) by \eqref{eq:ODE Picard Iteration}, eventually converge to a unique common prospective market $x$ (yellow line) satisfying \eqref{eq:ODEofFixedPoint}, and all three traders will agree to select the lower price $\theta_{1}$.  It is worth to note that even when the initial future prediction is strong (see the middle graph, Bob's initial prediction is $x_{0}(0+)=0$ and to select $\theta_{2}$), the trader corrects the prediction for a weaker one (to select $\theta_{1}$).  This is because the delay in traders’ reactions to the market always guarantees more traders selecting $\theta_{1}$ in short term, even when the most immediate initial subjective prediction tells the trader to select $\theta_{2}$.

Figure \ref{fig:ReversedPicardIteration} shows the case when the current market is below $p_{e}$ ($p_{0}=0.7 < p_{e}$).  Unlike the previous case, the limit prospect \eqref{eq:ODEofFixedPoint}, suggests the selection of $\theta_{2}$, even when the initial subjective prediction suggests the selection of $\theta_{1}$ in the future.

\begin{figure}[tbp]
\begin{center}
\includegraphics[width=0.45\textwidth]{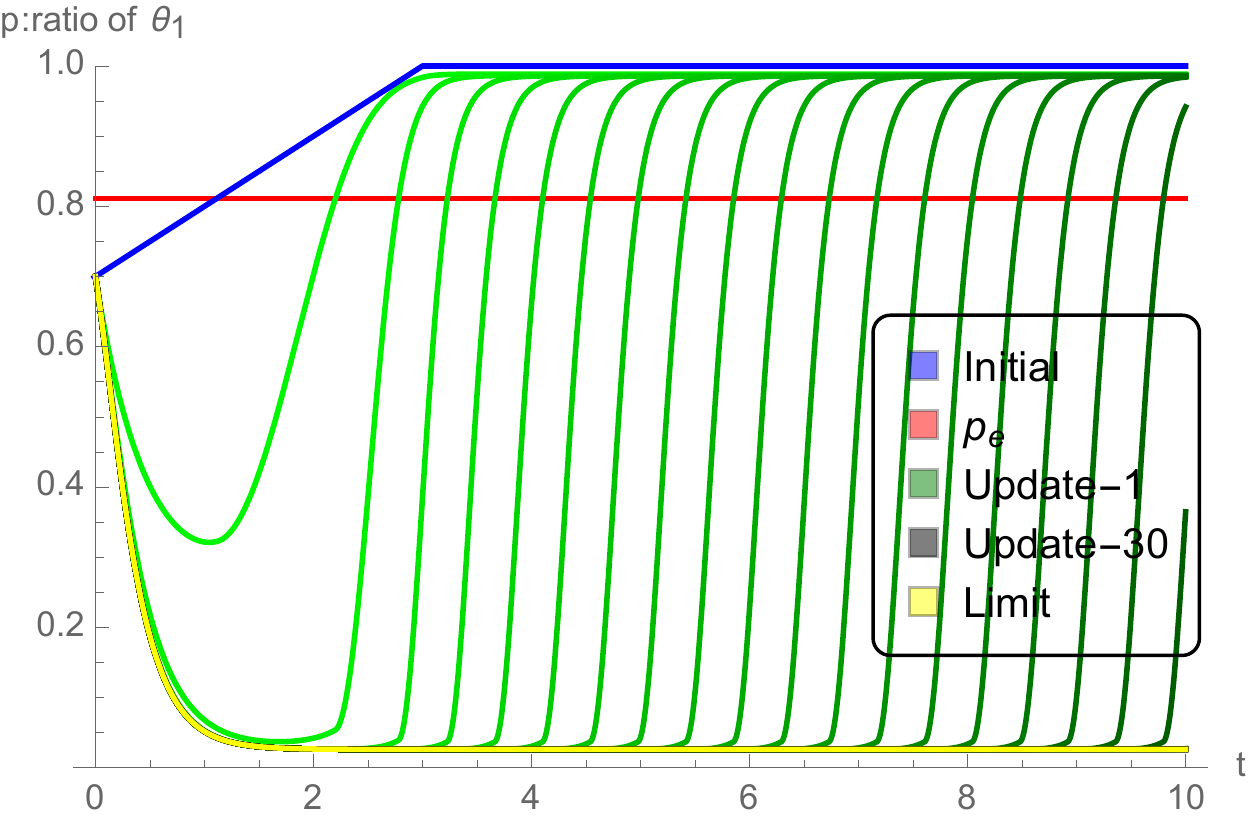}
\includegraphics[width=0.45\textwidth]{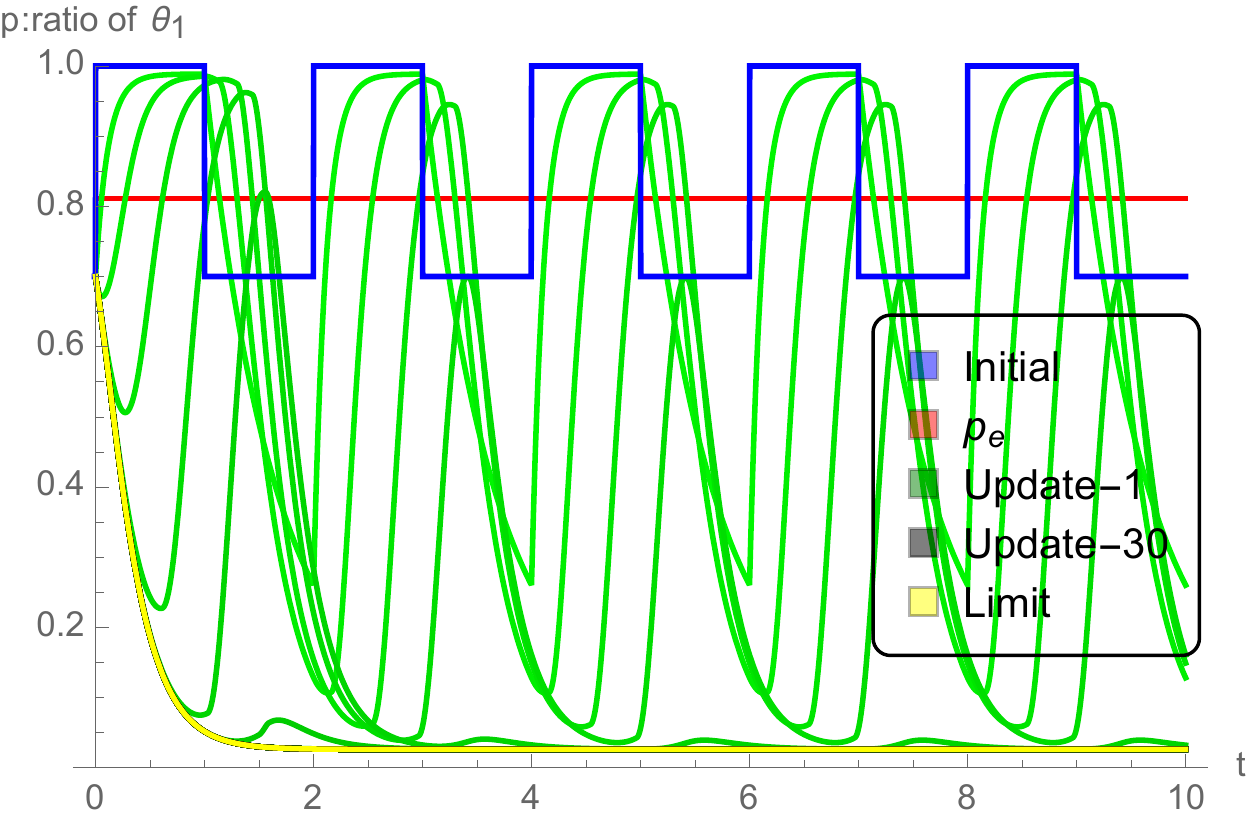}
\includegraphics[width=0.45\textwidth]{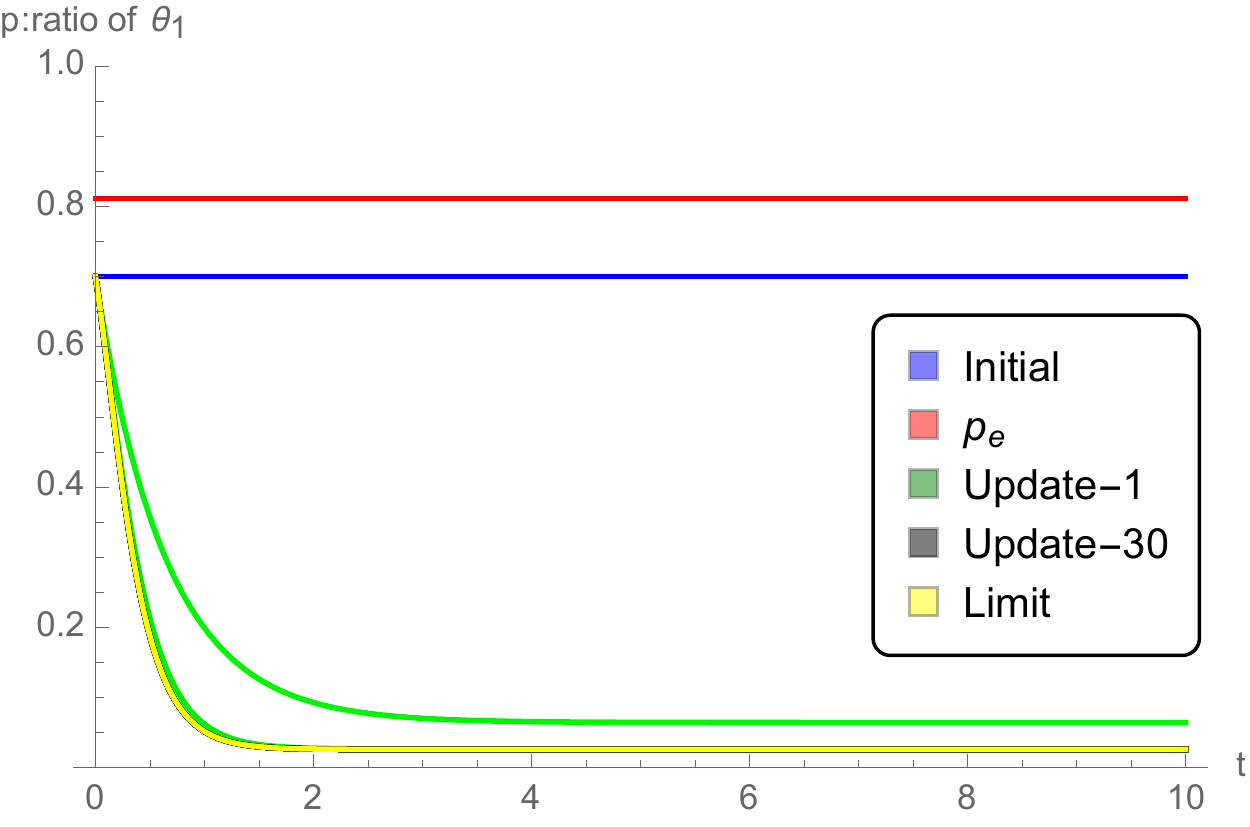}
\caption{Example of the market prediction starting from the current market ratio $p_{0}=0.7$.  The first 30 updates of predictions (green lines) are shown with the different initial predictions (blue lines).  Eventually all converge to the unique limit (yellow line) satisfying \eqref{eq:ODEofFixedPoint}. Here we set $\Delta \theta =1, \rho = 0.9, c = 0.03, \alpha = 5, \beta = 0.1$.  The critical ratio  is $p_{e} = 0.811111$.  }
\label{fig:ReversedPicardIteration}
\end{center}
\end{figure}

Figure \ref{fig:PicardIterationLOB-Equi.pdf} depicts the updates \eqref{eq:ODE Picard Iteration} in the market starting with the critical value $p_{e}$ and without random change ($\beta=0$).  In this case, although the current market is in the unstable equilibrium, the predictions converge to the unstable equilibrium.

\begin{figure}[tbp]
\begin{center}
\includegraphics[width=0.45\textwidth]{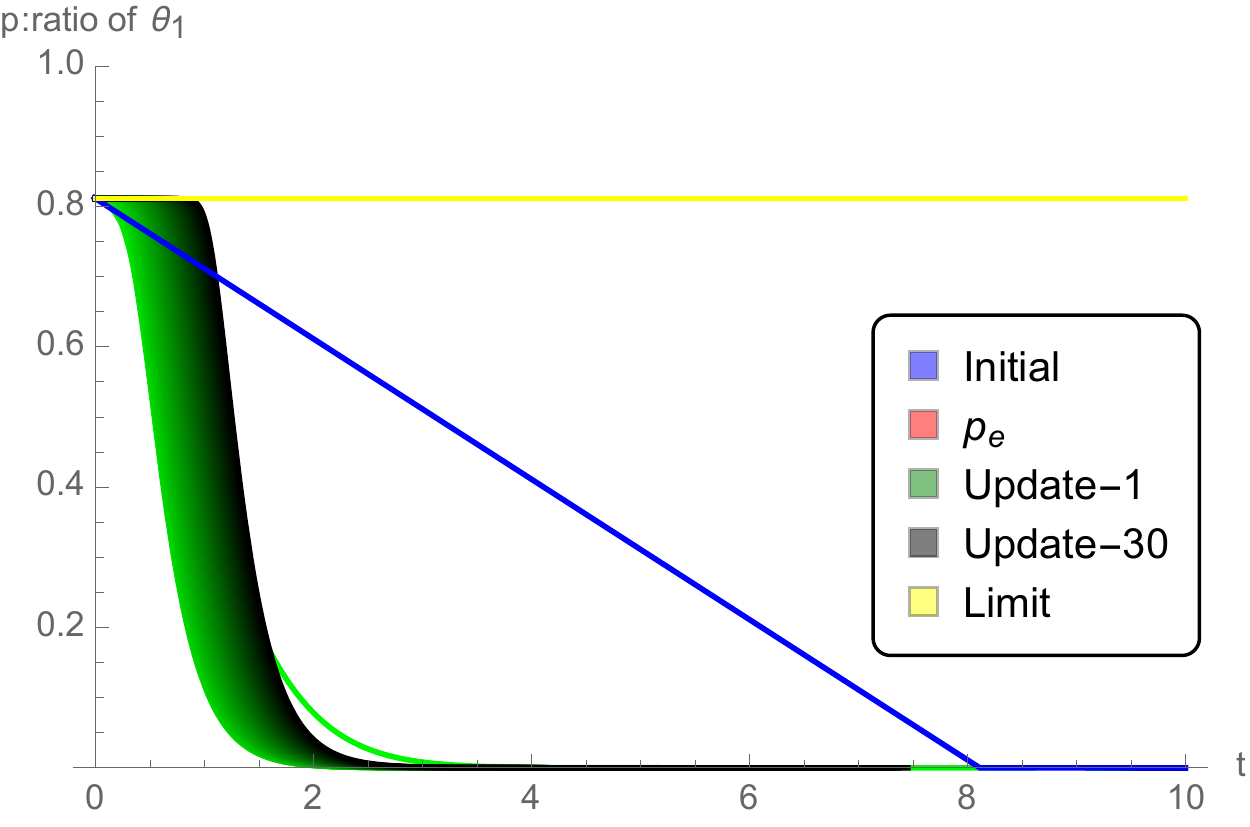}
\caption{Example of the market prediction by a trader starting from the current market ratio $p_{0}=p_{e}$.   The first 30 updates of predictions (green lines) are shown with the initial predictions (blue lines).  Here we set $\Delta \theta =1, \rho = 0.9, c = 0.03, \alpha = 5, \beta = 0$.    }
\label{fig:PicardIterationLOB-Equi.pdf}
\end{center}
\end{figure}

The random zero-intelligence behavior of traders is captured by $\beta$, which also influences the swarm limit $x$ (green lines with gradation from $\beta=0$ to $\beta=0.25$), as seen in Figure \ref{fig:PrandomEffect.pdf}. The smaller the random effect $\beta$, the more dependent on the initial value $p_{0}$.  Thus, the swarm limit $x$ depends on both the current market $p_{0}$ and the random behavior $\beta$.   The gradient to $\theta_{1}$ (the right hand side of \eqref{eq:ODEofFixedPoint}) at the time $0$ on the $(p_{0},\beta)$-plane is depicted in Figure \ref{fig:randomEffectContour.pdf}, which shows the complex response to the parameters.  This is because the random behavior of traders makes the system mean-reverting to $p=1/2$, while the boundaries $p=0$ and $p=1$ are stable equilibriums.  

Traditionally, the difference among traders' subjective predictions in the market is believed to randomize traders' behavior \cite{fama1965behavior}, however our analysis may suggest that traders' subjective predictions will eventually converges and the randomness of market is mainly caused by the zero-intelligence behaviors.

\begin{figure}[tbp]
\begin{center}
\includegraphics[width=0.45\textwidth]{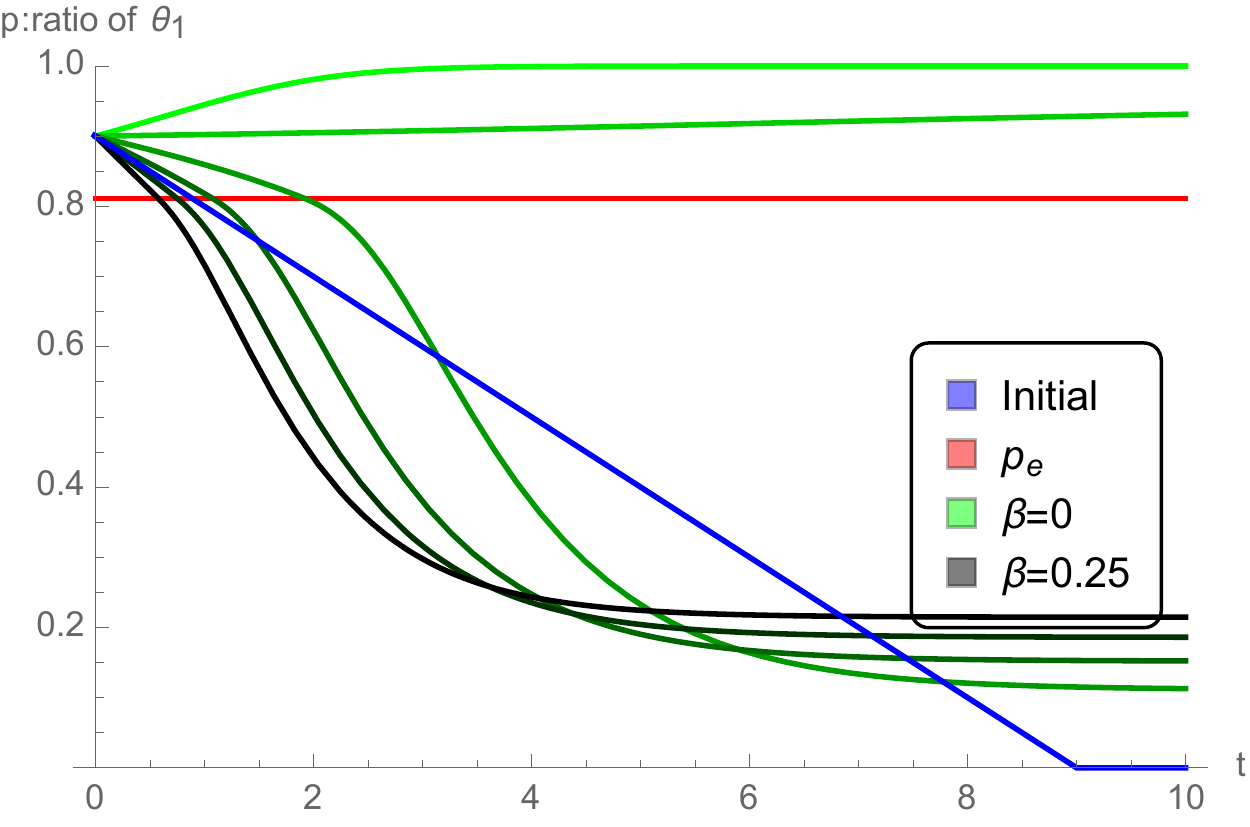}
\caption{Effect of zero-intelligent random switching $\beta$ starting with the current market ratio $p_{0}=0.9$.  The limits $x$ with the different $\beta=0,05,0.1,0.15, 0.2, 0.25$ are shown in the green lines with gradation.  Here we set $\Delta \theta =1, \rho = 0.9, c = 0.03, \alpha = 5$.}
\label{fig:PrandomEffect.pdf}
\end{center}
\end{figure}

\begin{figure}[tbp]
\begin{center}
\includegraphics[width=0.45\textwidth]{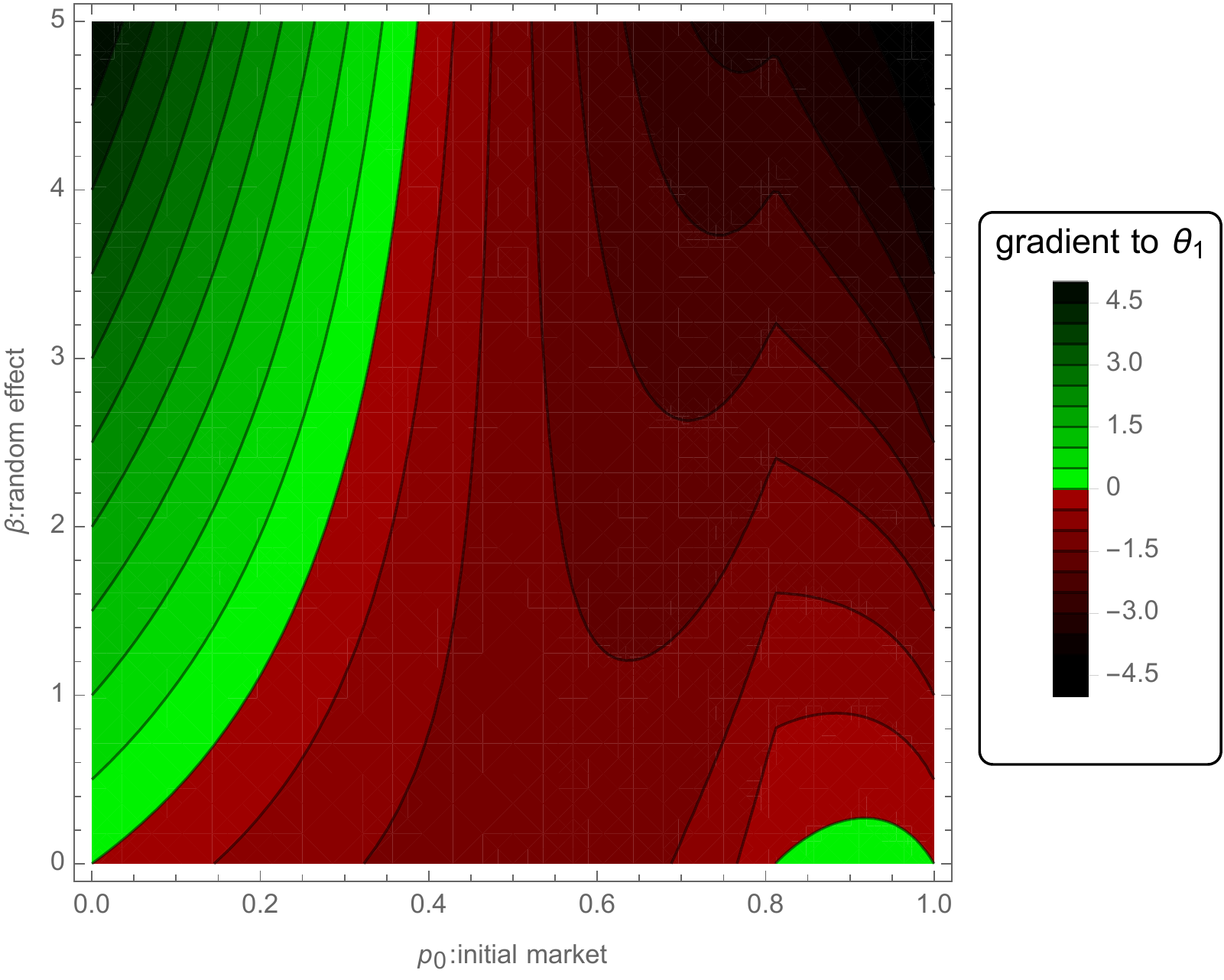}
\caption{Gradient to $\theta_{1}$ at the time $0$ ($-\alpha_{1}(p_{0})p_{0}+\alpha_{2}(p_{0})(1-p_{0})$) on $(p_{0},\beta)$-plane.  Here we set $\Delta \theta =1, \rho = 0.9, c = 0.03, \alpha = 5$.}
\label{fig:randomEffectContour.pdf}
\end{center}
\end{figure}

\bibliographystyle{apsrev4-1}
\bibliography{/Users/toyo/BoxSyncNoSpace/Public_Share/References/2015,/Users/toyo/BoxSyncNoSpace/Public_Share/References/2016,/Users/toyo/BoxSyncNoSpace/Public_Share/References/2017}

\clearpage

\section*{Supplemental Material}

\subsection{Expected Waiting Times $E[W_{1}]$ and $E[W_{2}]$}\label{Expected Waiting Times}
Here we summarize the known results for $M/M/1$ priority queues required for analyzing the LOB (see \citep{hassin2003queue,kleinrock,doi:10.1080/14697688.2014.963654}).

The waiting time of a $M/M/1$ queue with Poisson arrival (the rate $\lambda$) and the exponential service time (the rate $\mu$) is $1/(\mu- \lambda )$, when $\rho =\lambda/\mu <1$ and the queue is in the stationary state.  In the LOB model, the service time corresponds to the time required for the most favorable order to be executed, which equals to the inter-arrival time of buy market orders.

Let $p$ be the ratio of traders who place $\theta_{1}$ orders, and assume that the priority queue is in the stationary state.  Since low-prioritized $\theta_{2}$ orders do not affect the high-priority $\theta_{1}$ queue, the latter can be modeled by a $M/M/1$ queue with the arrival rate $\lambda p$.  Thus,  $E[W_{1}]=1/(\mu - \lambda p)$.  In addition, the aggregated mean waiting time of the $\theta_{1}$ and $\theta_{2}$ queues, which can be estimated by an another $M/M/1$ queue with the arrival rate $\lambda$, is $1/(\mu - \lambda)$, since the total waiting time (workload) is indifferent to whether or not the priority is adopted (the priority to one order is always compensated by the delay of others).  Thus, by the conservation law of workload, we have
\begin{align}
\frac{1}{\mu-\lambda} = p E[W_{1}] + (1-p) E[W_{2}].
\end{align}
Re-arranging the terms, we obtain the expected waiting time in $\theta_{2}$ as
\begin{align}
E[W_{2}] = \frac{\mu}{(\mu-\lambda)(\mu-\lambda p)}.
\end{align}

\subsection{Picard-like Iteration of Market Prediction}\label{Picard-like Iteration of Market Prediction}
Here we prove the uniform convergence of the Picard-like iteration $x_{n}$ to a unique function $x$ on the finite interval $[0,T]$ as in \eqref{eq:the limit}.

Let $g(p,c)$ be the function defined by
\begin{align}
g(p,c)=\frac{\rho c/\mu}{(1-\rho)(1-\rho p)} - \Delta \theta,
\end{align}
as in \eqref{eq:gain of 1}.  Then, $g(p,c)$ is Lipschitz continuous in $p\in[0,1]$.  Indeed, since $g(p,c)$ is convex and 
\begin{align}
\frac{\partial g}{\partial p}(p,c) \leq \frac{\partial g}{\partial p}(1,c) = \frac{\rho^{2} c/\mu}{(1-\rho)^{3}},
\end{align}
we have
\begin{align}
|g(p,c)-g(q,c)| \leq \frac{\rho^{2} c/\mu}{(1-\rho)^{3}}|p-q|.
\end{align}
Let $a(p)= -\left\{\alpha_{1}(p)+\alpha_{2}(p))\right\}= -2\beta-\alpha|g(p,c)|$ and $b(p)=\alpha_{2}(p)=\beta +\alpha g^{+}(p,c)$.  Since $g(p,c)$ is Lipchitz continuous and bounded function of $p\in[0,1]$, so are the functions $a(p)$ and $b(p)$.

Given an arbitrary measurable function $x(t)$ that has the value in $[0,1]$, consider a non-homogeneous differential equation for $u$:
\begin{align}\label{eq:non-homogeneous differential equation}
\frac{d}{dt}u(t) = a(x(t)) u(t) + b(x(t)),
\end{align}
with the initial condition $u(0)=p_{0}$, which is equivalent to the integral equation:
\begin{align}
u(t) = p_{0} e^{\int_{0}^{t}a(x(s))ds}+\int_{0}^{t}b(x(s))e^{\int_{s}^{t}a(x(s'))ds'}ds.
\end{align}

We give a useful lemma to evaluate the difference of the solutions of \eqref{eq:non-homogeneous differential equation}.
\begin{lemma}\label{lemma:difference}
Given two arbitrary functions $x(t)$ and $y(t)$, consider two differential equations:
\begin{align}
\frac{d}{dt}u(t) &= a(x(t)) u(t) + b(x(t)),\\
\frac{d}{dt}v(t) &= a(y(t)) u(t) + b(y(t)),
\end{align}
with the common initial condition $u(0)=v(0)=p_{0}$.  Then, there exists a constant $L>0$ such that for all $t \in [0,T]$,
\begin{align}
|u(t)-v(t)|\leq L\int_{0}^{t}|x(s)-y(s)|ds.
\end{align}
\end{lemma}

\begin{proof}
Since both $a(p)$ and $b(p)$ are Lipchitz continuous and bounded functions on $[0,1]$, we have
\begin{align}
\frac{d}{dt}|u(t)-v(t)|\leq L_{1}|u(t)-v(t)|+L_{2}|x(t)-y(t)|,
\end{align}
for some positive constant $L_{1}$ and $L_{2}$.  By Gronwall inequality, we have
\begin{align}
|u(t)-v(t)| &\leq |u(0)-v(0)|e^{L_{1}t} +\int_{0}^{t}L_{2}|x(s)-y(s)|e^{L_{1}(t-s)}ds\\
&\leq L\int_{0}^{t}|x(s)-y(s)|ds,
\end{align}
where $L=L_{2}e^{L_{1}T}$.
\end{proof}

Given an arbitrary initial measurable function $x_{0}(t)$ with the initial value $x_{0}(0)=p_{0}$, we define the Picard-like iteration $x_{n}$ by 
\begin{align}
x_{n}(t) =  p_{0} e^{\int_{0}^{t}a(x_{n-1}(s))ds}+\int_{0}^{t}b(x_{n-1}(s))e^{\int_{s}^{t}a(x_{n-1}(s'))ds'}ds,
\end{align}
which is equivalent to
\begin{align}
\frac{d}{dt}x_{n}(t) = a(x_{n-1}(t)) x_{n}(t) + b(x_{n-1}(t)),
\end{align}
with the initial condition $x_{n}(0)=p_{0}$.  Note that the original Picard iteration is derived from $dx_{n}(t)/dt = a(x_{n-1}(t)) x_{n-1}(t) + b(x_{n-1}(t))$ (see \cite{Dobrushkin:2014rt} for example), which is slightly different than ours.

Since both values $x_{0}(t)$ and $x_{1}(t)$ are always in $[0,1]$, $|x_{1}(t)-x_{0}(t)|\leq 1$. By Lemma \ref{lemma:difference},
\begin{align}
|x_{2}(t)-x_{1}(t)| \leq L\int_{0}^{t}|x_{1}(s)-x_{0}(s)|ds \leq Lt.
\end{align}
Inductively, we can show 
\begin{align}
|x_{n+1}(t)-x_{n}(t)| \leq\frac{(Lt)^{k}}{k!},
\end{align}
and for all $t \in [0,T]$,
\begin{align}
\sum_{n=0}^{\infty}|x_{n+1}(t)-x_{n}(t)| \leq e^{Lt}.
\end{align}
Hence, by Weierstrass M-test, the infinite sum $\sum_{n=0}^{\infty}(x_{n+1}(t)-x_{n}(t))$ converges uniformly on $[0,T]$.  Then, 
\begin{align}
x_{n}(t) &= x_{0}(t)+ \sum_{k=0}^{n}(x_{k+1}(t)-x_{k}(t)) \\
&\to x_{0}(t)+ \sum_{k=0}^{\infty}(x_{k+1}(t)-x_{k}(t)) =x(t),
\end{align}
uniformly on $[0,T]$ as $n\to\infty$.  The function $x(t)$ is well defined by this limit.  

By the uniform convergence of $x_{n}$ to $x$, we can deduce that the limit $x$ satisfies
\begin{align}
x(t) =  p_{0} e^{\int_{0}^{t}a(x(s))ds}+\int_{0}^{t}b(x(s))e^{\int_{s}^{t}a(x(s'))ds'}ds,
\end{align}
which is equivalent to 
\begin{align}
\frac{d}{dt}x(t) = a(x(t)) u(t) + b(x(t)),
\end{align}
with the initial condition $u(0)=p_{0}$.  This proved the existence of the limit $x$ starting from $x_{0}$.

Consider two different iterations $x_{n}$ and $y_{n}$ starting from two different functions $x_{0}$ and $y_{0}$ with the common initial condition: $x_{0}(0)=y_{0}(0)=p_{0}$. 
By the fact $|x_{0}(t)-y_{0}(t)|\leq 1$ and Lemma \ref{lemma:difference}, we have  
\begin{align}
|x_{1}(t)-y_{1}(t)| \leq L\int_{0}^{t}|x_{0}(s)-y_{0}(s)|ds \leq Lt.
\end{align}
Thus,  by the similar inductive arguments, we have the bound for the infinite sum:
\begin{align}
\sum_{n=0}^{\infty}|x_{n}(t)-y_{n}(t)|<e^{LT},
\end{align}
which suggests 
\begin{align}
|x_{n}(t)-y_{n}(t)|\to 0.
\end{align}
Thus,
\begin{align}
|x(t)-y(t)| \leq |x(t)-x_{n}(t)|+|x_{n}(t)-y_{n}(t)|+|y_{n}(t)-y(t)| \to 0. 
\end{align}
This shows that the limit of Picard iteration $x$ is independent of the choice of initial function $x_{0}$.

\end{document}